\documentclass[11pt]{article}
\usepackage{amsmath,amssymb,amsthm,natbib}
\usepackage{bm}
\usepackage{geometry}
\newtheorem{definition}{Definition}

\newtheorem{myproposition}{Proposition}
\newtheorem{assumption}{Assumption}

\newcommand{\comments}[1]{}
\title{Backtracking Direct, Indirect and Total Effects}
\author{Jose M. Pe\~{n}a\\STIMA, IDA, LiU}
\date{}
\begin{document}
\maketitle

\begin{abstract}%
Natural direct and indirect effects are defined through cross-world interventional counterfactuals. We propose an alternative construction based on backtracking counterfactuals. Rather than intervening on the treatment or mediator and change the structural equations, we use a backtracking conditional distribution to construct an alternative exogenous state satisfying a desired treatment-mediator constraint and then evaluate the unchanged structural equations on that state. This yields the backtracking direct and indirect effects, and also the backtracking total effect. Likewise, we adapt the population intervention effects. 
\end{abstract}

\section{Introduction}

Consider an employee who currently works in the office but experiences a major change in her background circumstances, e.g., caregiving responsibilities, commute difficulties, or home working conditions. In response, the employee switches to working from home. This change in work arrangement may affect productivity directly, as well as indirectly through its effect on the time spent collaborating with colleagues, which serves as a mediator. The company wants to assess the impact on the average productivity of such change in the circumstances.

The natural direct effect (NDE) addresses the question of what would happen to productivity if we intervene to change the employees' work arrangement while keeping the mediator at its value under the reference work arrangement and holding the background circumstances fixed \citep{pearl2009causality}. Although this intervention-based question is meaningful, it does not capture that the change in work arrangement is due to a change in the circumstances of its employees.

This motivates us to introduce in this paper the backtracking direct effect (BDE), which considers the change in productivity due to a change in work arrangement induced by a change in background circumstances, while the mediator remains unchanged. Whereas the NDE assesses the effect of an action, our BDE assesses the effect of a circumstances-driven switch. Our BDE thus provides a complementary perspective for causal questions concerning the response of individuals to changes in their circumstances.

An alternative motivation comes from the difference between interventional and backtracking counterfactuals \citep{voniigelgen2023backtracking}: the former examine how changing the structural equations affects the outcome while holding the exogenous state fixed, whereas the latter examine how changing the exogenous state affects the outcome while keeping the structural equations unchanged. The NDE is based on interventional counterfactuals, whereas our BDE is based on backtracking counterfactuals.

By being an adaptation of the NDE to backtracking counterfactuals, our BDE inherits features from both concepts. Like both, it presupposes a structural causal model. Like the NDE, it requires the choice of a reference treatment state. Like backtracking counterfactuals, it requires specifying a backtracking conditional distribution that determines which alternative exogenous states are possible relative to the factual one. Consequently, the BDE depends on these choices, just as its two conceptual antecedents depend on their respective choices.

We do not limit ourselves to the BDE, and we also propose the backtracking indirect effect (BIE) and the backtracking total effect (BTE). We also adapt the population intervention effects \citep{fulcher2020robust} to backtracking. The rest of the paper is organized as follows. First, we review interventional and backtracking counterfactuals, as well as natural and population intervention direct, indirect and total effects. Then, we introduce our backtracking adaptations thereof, including a numerical example. Finally, we relate the natural and backtracking effects, as well as the different backtracking effects among themselves.

\section{Counterfactuals and Direct, Indirect and Total Effects}\label{sec:review}

This section reviews counterfactuals and some measures of the direct, indirect and total effects. All random variables are discrete. Consider a structural causal model (SCM) with endogenous variables
\[
X=f_X(U_X),\qquad
M=f_M(X,U_M),\qquad
Y=f_Y(X,M,U_Y),
\]
and exogenous variables or background circumstances
\[
U=(U_X,U_M,U_Y)\sim P_U.
\]
Let
\[
X(U)=f_X(U_X),\qquad
M(U)=f_M(X(U),U_M),\qquad
Y(U)=f_Y(X(U),M(U),U_Y)
\]
and
\[
M_x=f_M(x,U_M),\qquad
Y_{x,M_{\bar{x}}}=f_Y(x,M_{\bar{x}},U_Y),\qquad
Y_x=Y_{x,M_x}.
\]
Moreover, let
\[
P_{BT}(u' \mid u)
\]
denote a backtracking conditional distribution that assigns a distribution over alternative or counterfactual exogenous states \(U'\) conditional on a factual state \(U=u\). Hereafter, a prime on an endogenous variable denotes the same variable with $U$ replaced by $U'$ in its definition, e.g., $X'=f_X(U_X')$ and $Y_{x,M_{\bar{x}}}'=f_Y(x,M_{\bar{x}},U_Y')$.

A counterfactual asks what the outcome would have been had something been different, for a unit with a given factual exogenous state. Interventional counterfactuals \citep{pearl2009causality} answer by evaluating the SCM modified by replacing some structural equations with constants, while holding the exogenous state fixed. For instance, the intervention $X=x$ implies replacing the equation $X(U)=f_X(U_X)$ with $X(U)=x$.

Interventional counterfactuals are sometimes criticized for relying on ``small miracles", i.e., overriding a structural equation amounts to violating the laws of nature. Backtracking counterfactuals \citep{voniigelgen2023backtracking} avoid this by attributing the difference to a change in the background circumstances rather than in the structural equations, and are therefore better suited to questions about how individuals respond to changed circumstances. Specifically, backtracking counterfactuals answer by evaluating the original SCM on an alternative exogenous state $u'$ drawn from $P_{BT}(u' \mid u)$, which favors states close in some sense to the factual one $u$. For a counterfactual condition $C(U,U')$ on the variables, e.g., $C(U,U')=\{X(U)=0, X(U')=1\}$, the backtracking counterfactual distribution of the outcome is
\[
P(Y(u')=y \mid u,C(u,u'))
=\frac{\sum_{u'} 1_{Y(u')=y}\,1_{C(u,u')}\,P_{BT}(u' \mid u)}
{\sum_{u'} 1_{C(u,u')}\,P_{BT}(u' \mid u)}.
\]

The natural direct, indirect and total effects \citep{pearl2009causality} for the contrast between the population with \(X=1\) and \(X=0\) are
\[
\text{NDE} = \mathbb{E}[Y_{1,M_0}-Y_{0,M_0}], \qquad \text{NIE} = \mathbb{E}[Y_{0,M_1}-Y_{0,M_0}], \qquad
\text{TE} = \mathbb{E}[Y_{1}-Y_{0}].
\]

The population intervention direct, indirect and total effects \citep{fulcher2020robust} for the contrast between the current population and one in which the exposure is set to a reference value $X=0$ are
\[
\text{PIDE}=\mathbb{E}[Y_{X,M_0} - Y_{0,M_0}], \qquad
\text{PIIE}=\mathbb{E}[Y_{X,M_X} - Y_{X,M_{0}}], \qquad
\text{PIE}=\mathbb{E}[Y_X - Y_{0}].
\]

The NDE, NIE and TE contrast treatment values $1$ and $0$ for the whole population, i.e., everybody treated against nobody treated. The PIDE, PIIE and PIE instead start from the population as it is, with each unit at its natural exposure $X(u)$, and move the exposure and/or the mediator to its value under $X(u)=0$. Units already at $X(u)=0$ contribute nothing, which makes these effects suited to harmful exposures such as smoking, where the question is what would happen if the exposure were eliminated.

Note that the effects described above are all defined through interventional counterfactuals.

\section{Backtracking Direct, Indirect and Total Effects}\label{sec:BEs}

This section proposes alternatives to the NDE, NIE and TE based on backtracking counterfactuals.

\subsection{Backtracking Direct Effect}\label{sec:BDE}

We propose an alternative to the NDE, in which the treatment changes by a backtracking operation rather than by an intervention, while the mediator remains unchanged. We present first an intuitive construction of our proposal, followed by its formal definition. Consider a unit $u$ for which
\[
X(u)=0.
\]
The factual mediator is then
\[
M(u)=f_M(0,u_M)=M_0.
\]
Draw an alternative exogenous state according to the backtracking conditional distribution
\[
U' \mid U=u
\sim
P_{BT}(\cdot \mid u)
\]
and retain the draw whenever
\[
X(u')=1
\qquad\text{and}\qquad
M(u')=M(u).
\]
Thus, the treatment changes to its alternative value while the mediator is constrained to remain at its factual natural value under the reference treatment. Finally, contrast \(Y(u)\) and \(Y(u')\). Importantly, the construction does not require \(U_Y'=U_Y\). The outcome disturbance is therefore allowed to change under backtracking.

\begin{definition}\label{def:bde}
Let $R_D(U,U')=\{X(U)=0, \,X(U')=1, \,M(U')=M(U)\}$ denote the retained \((U,U')\) pairs.\footnote{Throughout, we assume that the retention sets considered have positive probability, since otherwise the corresponding effect is undefined. This holds, for instance, if $P_U(u)>0$ and $P_{BT}(u'\mid u)>0$ for all $u,u'$, and $f_X(\cdot)$ and $f_M(x,\cdot)$ are onto for each $x$.} The backtracking direct effect associated with the backtracking conditional distribution \(P_{BT}\) is
\begin{align*}
\text{BDE}_{P_{BT}}
=
\frac{\sum_{(u,u')} \,[Y(u') - Y(u)] \,P_U(u) \,P_{BT}(u' \mid u) \,1_{R_D(u,u')}}{\sum_{(u,u')} \,P_U(u) \,P_{BT}(u' \mid u) \,1_{R_D(u,u')}}
=&
\mathbb{E}[
Y(U')-Y(U) | R_D(U,U')
]\\
=&
\mathbb{E}[
Y_{1, M_0}' - Y_{0, M_0} | R_D(U,U')
].
\end{align*}
\end{definition}

Note that the backtracking conditional distribution is part of the definition, so different distributions can produce different effects.

The BDE can be decomposed into a direct effect term and a backtracking term as
\[
\mathbb{E}[
Y_{1, M_0}' - Y_{0, M_0} \mid R_D(U,U')
] = \mathbb{E}[
Y_{1, M_0}' - Y_{0, M_0}' \mid R_D(U,U')
] + \mathbb{E}[
Y_{0, M_0}' - Y_{0, M_0} \mid R_D(U,U')
].
\]
The first term corresponds to the direct effect in the counterfactual world of changing the treatment while holding the mediator fixed at its factual natural value. The second term captures the change induced by the backtracking of the outcome disturbance, so the BDE and NDE need not coincide.

Note that the BDE considers the same treatment change and unchanged mediator as the NDE, but the change results from the alternative circumstances produced by backtracking rather than from an intervention. These alternative circumstances may also affect the outcome directly. The BDE is therefore useful for measuring the effect of the individuals' response to a change in their circumstances, rather than to an imposed action with the circumstances held fixed. In the introductory example, it tells the company what change in average productivity to expect if employees' background circumstances change and, in response, they switch to working from home while continuing to collaborate with colleagues as before.

\subsection{Backtracking Indirect Effect}\label{sec:BIE}

We also propose an alternative to the NIE, in which the mediator changes by a backtracking operation rather than by an intervention, while the treatment remains unchanged. We begin again with an intuitive construction of our proposal, followed by its formal definition. Start with a unit $u$ satisfying
\[
X(u)=0.
\]
The mediator value that the NIE considers is its natural value under the alternative treatment, $M_1=f_M(1,u_M)$. Draw
\[
U' \mid U=u \sim P_{BT}(\cdot \mid u)
\]
and retain the draw whenever
\[
X(u')=0 \quad\text{and}\quad M(u')=M_1.
\]
Thus, the treatment remains at its reference value, while the mediator is constrained to take the factual natural value under the alternative treatment. Finally, contrast $Y(u)$ and $Y(u')$. No restriction is imposed requiring $U'_Y=U_Y$.

\begin{definition}
Let $R_I(U,U')=\{X(U)=0,\,X(U')=0,\,M(U')=M_1\}$ denote the retained $(U,U')$ pairs. The backtracking indirect effect associated with the backtracking conditional distribution $P_{BT}$ is
\begin{align*}
\text{BIE}_{P_{BT}}
=
\frac{\sum_{(u,u')}[Y(u')-Y(u)]\,P_U(u)\,P_{BT}(u' \mid u)\,1_{R_I(u,u')}}
     {\sum_{(u,u')}P_U(u)\,P_{BT}(u' \mid u)\,1_{R_I(u,u')}}
=&
\mathbb{E}[Y(U')-Y(U) \mid R_I(U,U')]\\
=&
\mathbb{E}[Y'_{0,M_1}-Y_{0,M_0} \mid R_I(U,U')].
\end{align*}
\end{definition}

Note that the BIE depends on the specified backtracking conditional distribution. Moreover, the BIE can be decomposed into an indirect effect term and a backtracking term as
\[
\mathbb{E}[Y'_{0,M_1}-Y_{0,M_0} \mid R_I(U,U')] = \mathbb{E}[Y'_{0,M_1}-Y_{0,M_0}' \mid R_I(U,U')] + \mathbb{E}[Y'_{0,M_0}-Y_{0,M_0} \mid R_I(U,U')].
\]
The first term represents the effect in the counterfactual world of changing the factual natural value of the mediator while keeping the treatment fixed at its reference value. The second term captures the change induced by the backtracking of the outcome disturbance, so the BIE and NIE need not coincide.

Note that the BIE considers the same unchanged treatment and mediator change as the NIE, but the change results from the alternative circumstances produced by backtracking rather than from an intervention. These alternative circumstances may also affect the outcome directly. The BIE is therefore useful for measuring the effect of the individuals' response to a change in their circumstances, rather than to an imposed action with the circumstances held fixed. In the introductory example, it tells the company what change in average productivity to expect if employees' circumstances change in a way that leads them to collaborate with colleagues as much as they would have if they worked from home, but they keep working in the office.

\subsection{Backtracking Total Effect}\label{sec:BTE}

We finally propose an alternative to the TE, in which the treatment changes by a backtracking operation rather than by an intervention. We begin again with an intuitive construction of our proposal, followed by its formal definition. Consider a unit $u$ satisfying $X(u)=0$. Draw an alternative unit according to
\[
U' \mid U=u \sim P_{BT}(\cdot \mid u),
\]
and retain it whenever $X(u')=1$. The total outcome contrast is then $Y(u')-Y(u)$. No restriction is imposed requiring $U'_M=U_M$ or $U_Y'=U_Y$.

\begin{definition}\label{def:bte}
Let $R_T(U,U')=\{X(U)=0, \,X(U')=1\}$ denote the retained \((U,U')\) pairs. The backtracking total effect associated with the backtracking conditional distribution $P_{BT}$ is
\begin{align*}
\text{BTE}_{P_{BT}}
=
\frac{\sum_{(u,u')} \,[Y(u') - Y(u)] \,P_U(u) \,P_{BT}(u' \mid u) \,1_{R_T(u,u')}}
{\sum_{(u,u')} \,P_U(u) \,P_{BT}(u' \mid u) \,1_{R_T(u,u')}}
=&
\mathbb{E}[
Y(U')-Y(U) | R_T(U,U')
]\\
=& 
\mathbb{E}[
Y_{1}'
-
Y_{0} \mid R_T(U,U')
].
\end{align*}
\end{definition}

Note that the BTE depends on the specified backtracking conditional distribution. Moreover, the BTE can be decomposed into a total effect term and a backtracking term as
\[
\mathbb{E}[
Y_{1}'
-
Y_{0}
 \mid R_T(U,U')
] = \mathbb{E}[
Y_{1}'
-
Y_{0}'
 \mid R_T(U,U')
] + \mathbb{E}[
Y_{0}'
-
Y_{0}
 \mid R_T(U,U')
].
\]
The first term is the total effect in the counterfactual world of changing the treatment. The second term is the change induced by backtracking the mediator and outcome disturbances, so the BTE and TE need not coincide.

Note that the BTE considers the same treatment change as the TE, but the change results from the alternative circumstances produced by backtracking rather than from an intervention. These alternative circumstances may also affect the outcome directly. The BTE is therefore useful for measuring the effect of the individuals' response to a change in their circumstances, rather than to an imposed action with the circumstances held fixed. In the introductory example, it tells the company what change in average productivity to expect if employees' circumstances change and, in response, they switch to working from home and adapt their collaboration with colleagues accordingly.

\subsection{Dependence on the Backtracking Conditional Distribution}

As noted before, the backtracking effects depend on the backtracking conditional probability distribution (see Appendix \ref{app:example} for an example). Some desiderata for this distribution have been discussed elsewhere \citep{voniigelgen2023backtracking}. Here, we complement that discussion and propose to regard the backtracking conditional distribution as a policy that determines how the current circumstances will change. The dependence of the backtracking effects on the policy is a feature, not a limitation, since it lets the analyst assess the effects of candidate policies before choosing any of them. Each backtracking effect is the effect of the policy on a specific subpopulation, namely the units retained by the corresponding retention set. This can guide the policy selection, as we explain next.

Consider the company of the introductory example, and suppose it can choose between policies that change its employees' circumstances, e.g., subsidizing home offices, offering flexible childcare, or changing commuting benefits. Each policy induces a backtracking conditional distribution. The three backtracking effects then answer the following questions for each policy:
\begin{itemize}
\item BTE: among the employees who switch to working from
home, what is the expected change in productivity, including the effect of their changed circumstances on collaboration and productivity?
\item BDE: among the employees who switch to working from home and collaborate with colleagues as before, what is the expected change in productivity?
\item BIE: among the employees who keep working in the office
but whose collaboration changes to what it would be under working from home, what is the expected change in productivity?
\end{itemize}
Further subpopulations are obtained by changing the reference value. With reference value $1$, the effects concern the employees currently working from home: those who would switch to the office, those who would switch and keep their collaboration, and those who would stay at home but collaborate as they would in the office. The company may want to know these backtracking effects for each policy to guide the choice, since the policies may differ in who responds and how.

However, each backtracking effect is an average over the retained individuals only. Two policies may retain different fractions of the employees and may affect the remaining employees differently, and so these effects alone do not provide a full comparison across policies. They are nevertheless informative for a company that cares about the response of particular subpopulations.

\section{Backtracking Population Direct, Indirect and Total Effects}\label{sec:BPEs}

The BDE, BIE and BTE concern only the employees who currently work in the office, i.e., $X(u)=0$. A company may instead care about its whole workforce as it is today, with some employees in the office and some already working from home. Suppose the company suspects that working from home is harmful to productivity. Its question is then what would happen if this exposure were eliminated. This is the situation for which the PIDE, PIIE and PIE of Section \ref{sec:review} were introduced, and in this section we adapt them to backtracking counterfactuals. That is, the exposure is not eliminated by an intervention but by a change in the circumstances.

\subsection{Backtracking Population Direct Effect}

Consider a unit $u$ with any exposure $X(u)$. Draw
\[
U' \mid U=u\sim P_{BT}(\cdot \mid u)
\]
and retain the draw whenever
\[
X(u')=0 \quad\text{and}\quad M(u)=M(u')=M_0(u').
\]
Thus, the exposure is eliminated, while the mediator is constrained to be fixed at its counterfactual natural value. Finally, contrast $Y(u)$ and $Y(u')$. As before, the outcome disturbance is allowed to change under backtracking.

\begin{definition}
Let $R_{PD}(U,U')=\{X(U')=0,\,M(U')=M(U)\}$ denote the retained $(U,U')$ pairs. The backtracking population direct effect
associated with $P_{BT}$ is
\[
\text{BPDE}_{P_{BT}}
=\mathbb{E}[Y(U)-Y(U') \mid R_{PD}(U,U')]
=\mathbb{E}[Y_{X,M_0'}-Y'_{0,M_0'} \mid R_{PD}(U,U')]
\]
where the expectation is the weighted average with weights
$P_U(u)P_{BT}(u' \mid u)1_{R_{PD}(u,u')}$, normalized as in Definitions \ref{def:bde}-\ref{def:bte}.
\end{definition}

The BPDE can be decomposed into a direct effect term and a backtracking term as
\[
\mathbb{E}[Y_{X,M_0'}-Y'_{0,M_0'} | R_{PD}(U,U')]
=\mathbb{E}[Y'_{X,M_0'}-Y'_{0,M_0'} | R_{PD}(U,U')]
+\mathbb{E}[Y_{X,M_0'}-Y'_{X,M_0'} | R_{PD}(U,U')].
\]
The first term is the direct effect in the counterfactual world of eliminating the exposure while holding the mediator fixed at its counterfactual natural value. The second term captures the change induced by the backtracking of the outcome disturbance, so the BPDE and PIDE need not coincide. Moreover, note that unlike the PIDE, to which the units with $X(u)=0$ contribute nothing, the BPDE may receive contributions from these units through the second term.

Note that the BPDE considers the same exposure change and unchanged mediator as the PIDE, but the change results from the alternative circumstances produced by backtracking rather than from an intervention. In the introductory example, it tells the company the change in average productivity, relative to today, to expect if the employees' circumstances change and, in response, they work in the office while collaborating with colleagues as before.

\subsection{Backtracking Population Indirect Effect}

Consider a unit $u$ with any exposure $X(u)$. The mediator
value that the PIIE considers is its natural value under the reference exposure, $M_0(u)=f_M(0,u_M)$. Draw $U' \mid U=u\sim P_{BT}(\cdot \mid u)$ and retain the draw whenever
\[
X(u')=X(u)\quad\text{and}\quad M(u')=M_0(u).
\]
Thus, the exposure remains at its factual value, while the mediator is constrained to take its factual natural value under the reference exposure. Finally, contrast $Y(u)$ and $Y(u')$. No restriction is imposed requiring $U'_Y=U_Y$.

\begin{definition}
Let $R_{PI}(U,U')=\{X(U')=X(U),\,M(U')=M_0(U)\}$ denote the retained $(U,U')$ pairs. The backtracking population indirect effect associated with $P_{BT}$ is
\[
\text{BPIE}_{P_{BT}}
=\mathbb{E}[Y(U)-Y(U') \mid R_{PI}(U,U')]
=\mathbb{E}[Y_{X,M_X}-Y'_{X,M_0} \mid R_{PI}(U,U')]
\]
where the expectation is the weighted average with weights
$P_U(u)P_{BT}(u' \mid u)1_{R_{PI}(u,u')}$, normalized as in Definitions \ref{def:bde}-\ref{def:bte}.
\end{definition}

The BPIE can be decomposed into an indirect effect term and a backtracking term
as
\[
\mathbb{E}[Y_{X,M_X}-Y'_{X,M_0} | R_{PI}(U,U')]
=\mathbb{E}[Y'_{X,M_X}-Y'_{X,M_0} | R_{PI}(U,U')]
+\mathbb{E}[Y_{X,M_X}-Y'_{X,M_X} | R_{PI}(U,U')].
\]
The first term is the effect in the counterfactual world of changing the factual mediator to its natural value under the reference exposure while keeping the exposure fixed. The second term captures the change induced by the backtracking of the outcome disturbance, so the BPIE and PIIE need not coincide. Moreover, note that unlike the PIIE, to which the units with $X(u)=0$ contribute nothing, the BPIE may receive contributions from these units through the second term.

Note that the BPIE considers the same unchanged exposure and mediator change as the PIIE, but the change results from the alternative circumstances produced by backtracking rather than from an intervention. In the introductory example, it tells the company the change in average productivity, relative to today, to expect if the employees' circumstances change in a way that leads them to collaborate as much as they would have if they worked in the office, while they keep their current work arrangement.

\subsection{Backtracking Population Total Effect}

Consider a unit $u$ with any exposure $X(u)$, draw $U' \mid U=u\sim P_{BT}(\cdot \mid u)$, and retain it whenever $X(u')=0$. The total outcome contrast is then $Y(u)-Y(u')$. No restriction is imposed requiring $U'_M=U_M$ or $U'_Y=U_Y$.

\begin{definition}
Let $R_{PT}(U,U')=\{X(U')=0\}$ denote the retained $(U,U')$ pairs. The backtracking population total effect associated with $P_{BT}$ is
\[
\text{BPTE}_{P_{BT}}
=\mathbb{E}[Y(U)-Y(U') \mid R_{PT}(U,U')]
=\mathbb{E}[Y_X-Y'_0 \mid R_{PT}(U,U')]
\]
where the expectation is the weighted average with weights
$P_U(u)P_{BT}(u' \mid u)1_{R_{PT}(u,u')}$, normalized as in Definitions \ref{def:bde}-\ref{def:bte}.
\end{definition}

The BPTE can be decomposed into a total effect term and a backtracking term as
\[
\mathbb{E}[Y_X-Y'_0 \mid R_{PT}(U,U')]
=\mathbb{E}[Y'_X-Y'_0 \mid R_{PT}(U,U')]
+\mathbb{E}[Y_X-Y'_X \mid R_{PT}(U,U')].
\]
The first term is the total effect in the counterfactual world of eliminating the exposure. The second term is the change induced by backtracking the mediator and outcome disturbances, so the BPTE and PIE need not coincide. Moreover, note that unlike the PIE, to which the units with $X(u)=0$ contribute nothing, the BPTE may receive contributions from these units through the second term.

Note that the BPTE considers the same exposure change as the PIE, but the change results from the alternative circumstances produced by backtracking rather than from an intervention. In the introductory example, it tells the company the change in average productivity, relative to today, to expect if the employees' circumstances change and, in response, they work in the office and adapt their collaboration with colleagues accordingly.

\section{Relations between Effects}

In this section, we relate the natural and backtracking effects, as well as the different backtracking effects among themselves.

\subsection{Backtracking Circumstances Effect}

In Section \ref{sec:BEs}, the effects are those of a change in the treatment or the mediator that is induced by backtracking, and the retention sets select the counterfactual worlds in which that change takes place. In this section, we take the change in circumstances itself as the object of study. We ask how much of its impact on the outcome goes through each of these three channels: induced treatment change, induced mediator change, and outcome disturbance change. This is useful in our introductory example because, for a policy that changes the employees' circumstances, the three channels tell the company how much of the expected change in productivity for the employees currently working in the office comes from changing where they work, from changes in how they collaborate, and directly from the new circumstances. A policy such as a commuting subsidy may act mostly through the first channel, whereas arranging team-building activities may act mostly through the second.

Consider a factual unit $u$ with $X(u)=0$ and draw
\[
U' \mid U=u\sim P_{BT}(\cdot \mid u)
\]
without discarding any draw. The impact of the change in circumstances on the outcome of this unit is $Y(u')-Y(u)$.

\begin{definition}
Let $R_C(U,U')=\{X(U)=0\}$ denote the retained $(U,U')$ pairs, i.e., no draw of $U'$ is discarded. The backtracking circumstances effect associated with $P_{BT}$ is
\[
\text{BCE}_{P_{BT}}
=\mathbb{E}[Y(U')-Y(U) \mid R_C(U,U')]
=\mathbb{E}[Y(U')-Y_0 \mid R_C(U,U')]
\]
where the expectation is the weighted average with weights
$P_U(u)P_{BT}(u' \mid u)1_{R_C(u,u')}$, normalized as in Definitions \ref{def:bde}-\ref{def:bte}.
\end{definition}

The BCE can be decomposed into three terms corresponding to the three channels of influence as follows
\[
\text{BCE}_{P_{BT}}=\text{BCE}^{\,X}+\text{BCE}^{\,M}+\text{BCE}^{\,U_Y}
\]
where
\begin{align*}
\text{BCE}^{\,X}&=\mathbb{E}[Y_{X',M'}'-Y_{X,M'}' \mid R_C(U,U')],\\
\text{BCE}^{\,M}&=\mathbb{E}[Y_{X,M'}'-Y_{X,M}' \mid R_C(U,U')],\\
\text{BCE}^{\,U_Y}&=\mathbb{E}[Y_{X,M}'-Y_{X,M} \mid R_C(U,U')].
\end{align*}
The first term is the effect of changing the treatment with the mediator at its counterfactual value. The second is the effect of changing the mediator with the treatment at its factual value. The third is the effect of changing the outcome disturbance with the treatment and mediator at their factual values, i.e., the pure effect of the new circumstances on the outcome that does not run through the treatment or mediator.

The BDE, BIE and BTE can be seen as special cases of the decomposition above. The BCE averages the contrast $Y(U')-Y(U)$ over all draws of $U'$. The BDE, BIE and BTE average the same contrast over subsets of these draws, namely $R_D(U,U')$, $R_I(U,U')$ and $R_T(U,U')$, respectively. Hence, the three-channel decomposition above also holds when the expectations are conditioned on any of these retention sets, and the retention set determines which channels are active:
\begin{itemize}
\item On $R_D(U,U')=\{X(U)=0,\,X(U')=1,\,M(U')=M(U)\}$, the mediator channel vanishes, since $M'=M$. The two remaining terms are the two terms of the decomposition of the BDE in Section \ref{sec:BDE}.
\item On $R_I(U,U')=\{X(U)=0,\,X(U')=0,\,M(U')=M_1\}$, the treatment channel vanishes, since $X'=X$. The two remaining terms are the two terms of the decomposition of the BIE in Section \ref{sec:BIE}. 
\item On $R_T(U,U')=\{X(U)=0,\,X(U')=1\}$, all three channels appear. The BTE is therefore not the sum of the BDE and BIE, but it is the sum of the three channel terms. This is an alternative decomposition to the one in Section \ref{sec:BTE}.
\end{itemize}

Note that the BCE is conditioned on the factual work arrangement $X(U)=0$ as in Section \ref{sec:BEs}, and so it concerns the employees currently working in the office. For the whole workforce of Section \ref{sec:BPEs}, the result holds on $R_{PD}(U,U')$, $R_{PI}(U,U')$ and $R_{PT}(U,U')$ with the contrast $Y(U)-Y(U')$.

\subsection{Agreement between Natural and Backtracking Effects}

As mentioned before, the natural and backtracking effects may differ (see Appendix \ref{app:example} for an example). However, we believe it is important to identify cases where the two coincide. The goal of this section is to describe one such case.

A valid question is why the counterfactual exogenous state $U'$ should have a different marginal distribution from the factual state $U$. A backtracking conditional distribution may describe a policy that changes the circumstances of the individuals and, consequently, shifts the distribution of circumstances in the population. One may instead wish to study changes in circumstances that preserve the population distribution while allowing the factual and counterfactual states to be dependent. We do so in this section. We additionally assume a componentwise coupling between $U$ and $U'$, and a binary and uniform $U_M$. These assumptions describe the case studied in this section. We formalize these assumptions next.

\begin{assumption}\label{ass:marginal}
The marginal distribution of $U'$ induced by $P_U$ and $P_{BT}$ equals $P_U$, i.e.,
\[
\sum_u P_{BT}(u' \mid u)\,P_U(u) = P_U(u')
\qquad\text{for all } u'.
\]
\end{assumption}

\begin{assumption}\label{ass:componentwise}
There is a componentwise coupling between $U$ and $U'$, i.e.,
\[
P_U(u)=P(u_X) P(u_M) P(u_Y)
\]
and
\[
P_{BT}(u' \mid u)=P(u'_X \mid u_X)\,P(u'_M \mid u_M)\,P(u'_Y \mid u_Y).
\]
\end{assumption}

\begin{assumption}\label{ass:uniform}
The mediator $M$ and its disturbance $U_M$ are binary, $U_M$ is uniform, and $f_M(x,\cdot)$ is a bijection for each $x$.
\end{assumption}

Under Assumptions \ref{ass:marginal} and \ref{ass:componentwise}, each component of $U'$ has the same marginal distribution as the corresponding component of $U$, and the pairs $(U_X,U'_X)$, $(U_M,U'_M)$ and $(U_Y,U'_Y)$ are independent of each other. Within each pair, the factual and counterfactual components may be dependent. Assumption \ref{ass:marginal} holds, for instance, if the joint distribution of $(U,U')$ is symmetric, i.e., $P_U(u)P_{BT}(u' \mid u)=P_U(u')P_{BT}(u \mid u')$, which may be considered a desirable property to avoid a priori asymmetries between the factual and counterfactual world \citep{voniigelgen2023backtracking}.\footnote{Lemma 4 of \citep{voniigelgen2023backtracking} states that Assumption \ref{ass:marginal} is equivalent to the symmetry $P_{BT}(u' \mid u) = P_{BT}(u \mid u')$, but this does not hold in general. For instance, let $U$ be binary with $P_U(0)=0.8$ and $P_{BT}(u' \mid u)=1/2$ for all $u,u'$. Then, $P_{BT}$ is symmetric but $\sum_u P_{BT}(u' \mid u)P_U(u)=1/2 \neq P_U(u')$, so Assumption \ref{ass:marginal} fails.} Likewise, Assumption \ref{ass:componentwise} may be considered a natural property in those domains where the factual exogenous variables are mutually independent \citep{voniigelgen2023backtracking}.

The following proposition is proved in Appendix \ref{app:relation}.

\begin{myproposition}\label{pro:relation}
Under Assumptions \ref{ass:marginal}-\ref{ass:uniform},
\[
\text{BDE}_{P_{BT}}=\text{NDE},\qquad
\text{BIE}_{P_{BT}}=\text{NIE},\qquad
\text{BTE}_{P_{BT}}=\text{TE}.
\]
\end{myproposition}

\subsection{Natural Effects as Backtracking Effects of an Augmented Model}\label{sec:augmented}

Natural effects can be expressed as backtracking effects in an augmented model. Specifically, inspired by \cite{pearl2009causality}, we augment the original SCM with two exogenous regime variables, $I_X\in\{\varnothing\}\cup\mathcal{X}$ and
$I_M\in\{\varnothing\}\cup\mathcal{M}$, where $\mathcal{X}$ and $\mathcal{M}$ are the sets of possible values of $X$ and $M$:
\[
X=\begin{cases} I_X & \text{if } I_X\neq\varnothing,\\ f_X(U_X) & \text{otherwise,}\end{cases}
\qquad
M=\begin{cases} I_M & \text{if } I_M\neq\varnothing,\\ f_M(X,U_M) & \text{otherwise,}\end{cases}
\qquad
Y=f_Y(X,M,U_Y).
\]

The exogenous state of the augmented model is $V=(U,I_X,I_M)$, and the original model is recovered in the idle state $V=(U,\varnothing,\varnothing)$. Since all structural equations are left untouched, any alternative state $V'$ defines a backtracking counterfactual of the augmented model, and the corresponding effect is $\mathbb{E}[Y(V')-Y(V) \mid R(V,V')]$. The backtracking conditional distribution $P(v'\mid v)$ then distinguishes the two families of effects: natural effects use a degenerate distribution that keeps $U$ and changes the regimes, whereas the backtracking effects use a distribution that changes $U$ and keeps the regimes idle. Specifically, $V'$ can be obtained from $V$ in two ways:
\begin{itemize}
\item changing the regimes while leaving $U$ unchanged, i.e., $V'=(U,x,\varnothing)$, $V'=(U,\varnothing,m)$ or $V'=(U,x,m)$ with $x \in \mathcal{X}$ and $m \in \mathcal{M}$, and where $m$ may depend on $U$, e.g., $m=M_0$,
\item changing $U$ while leaving the regimes idle, i.e., $V'=(U',\varnothing,\varnothing)$ with $U'\sim P_{BT}(\cdot \mid u)$.
\end{itemize}

The natural effects of Section \ref{sec:review} correspond to the first type of backtracking effects in the augmented SCM. Specifically, the following proposition is proved in Appendix \ref{app:unification}.

\begin{myproposition}\label{pro:unification}
Suppose that $U_X$ is independent of $(U_M,U_Y)$, i.e., that the treatment is unconfounded. Then
\begin{align*}
\text{TE}&=\mathbb{E}[Y(V')-Y(V) \mid R(V,V')] \text{ with } R(V,V')=\{X(V)=0\} \text{ and } V'=(U,1,\varnothing),\\
\text{NDE}&=\mathbb{E}[Y(V')-Y(V) \mid R(V,V')] \text{ with } R(V,V')=\{X(V)=0\} \text{ and } V'=(U,1,M(V)),\\
\text{NIE}&=\mathbb{E}[Y(V')-Y(V) \mid R(V,V')] \text{ with } R(V,V')=\{X(V)=0\} \text{ and } V'=(U,\varnothing,M_1).
\end{align*}
\end{myproposition}

On the other hand, the backtracking effects of Section \ref{sec:BEs} correspond to the second type of backtracking effects in the augmented SCM:
\begin{align*}
\text{BDE}_{P_{BT}}&=\mathbb{E}[Y(V')-Y(V) \mid R(V,V')] \text{ with } R(V,V')=\{X(V)=0,X(V')=1,M(V')=M(V)\},\\
\text{BIE}_{P_{BT}}&=\mathbb{E}[Y(V')-Y(V) \mid R(V,V')] \text{ with } R(V,V')=\{X(V)=0,X(V')=0,M(V')=M_1\},\\
\text{BTE}_{P_{BT}}&=\mathbb{E}[Y(V')-Y(V) \mid R(V,V')] \text{ with } R(V,V')=\{X(V)=0,X(V')=1\},
\end{align*}
where $V'=(U',\varnothing,\varnothing)$ and $U'\sim P_{BT}(\cdot \mid u)$.

Finally, the augmented SCM also allows counterfactual states obtained by hybrid changes, which alter both the circumstances and the regimes. See Appendix \ref{app:hybrid} for more on this.

\section{Discussion}

Interventional effects create a counterfactual counterpart for each factual individual by overriding the mechanism, and keep it. Backtracking effects create a counterfactual counterpart for each factual individual by changing the circumstances, and keep only those that meet the retention conditions, so they describe only the retained counterfactual subpopulations. The two kinds of effects therefore complement each other, because they address different questions about different (sub)populations. Neither answers the other’s question, and which one is appropriate depends on whether the interest lies in an action applied to the whole population or in the response of a subpopulation to new circumstances. In the backtracking effects, the cross-world coupling is induced by the backtracking conditional distribution, which is therefore part of the counterfactual specification of the estimands rather than merely a computational device. Different distributions can be seen as different policies, each describing how the circumstances of the individuals would change and, consequently, who would respond and how. The backtracking effects then help us select among such policies by showing their effects on the subpopulations of interest.

We conclude with two possible extensions of our work. In this paper, we have given some guidelines to select a single backtracking conditional distribution. An alternative is not to select one, but to report an interval. Given a set of admissible distributions, one can report the smallest and largest values of each backtracking effect over the set. This is a form of sensitivity analysis: a narrow interval means that the conclusion does not depend on how the circumstances change, while a wide one signals that the choice of the distribution matters.

In this paper, we have limited our presentation to discrete mediator and outcome. An extension to continuous mediator and outcome is possible. With a continuous mediator, requiring the mediator to stay exactly the same (B(P)DE) or to reach exactly its natural value under the other treatment (B(P)IE) is too strict: such an event has probability zero, so no alternative state would ever be retained. A straightforward fix is to be less demanding and retain alternative states in which the mediator is within a tolerance of the target value. A continuous outcome, in contrast, is unproblematic, since no retention constraint involves the outcome: it only enters the contrast being averaged. The B(P)TE needs no tolerance when the treatment is binary.

\bibliographystyle{plain}
\bibliography{references}

\appendix

\section{Example}\label{app:example}

We illustrate the computation of the natural and backtracking effects with an example. Specifically, consider the SCM
\[
U_X, U_M, U_Y \overset{\text{i.i.d.}}{\sim} \text{Bernoulli}\big(\tfrac{1}{2}\big),
\]
with structural equations
\[
X = U_X, \qquad M = X \oplus U_M, \qquad Y = X + M + U_Y,
\]
where $\oplus$ denotes addition modulo 2. Note that $1 \oplus u = 1 - u$ for $u \in \{0,1\}$. Consequently,
\[
Y_{x,m} = x + m + U_Y, \qquad M_0 = U_M, \qquad M_1 = 1 - U_M.
\]

For computing the backtracking effects, the backtracking conditional distribution $P_{BT}$ is only needed for factual states with $X(U)=0$, i.e.\ $U_X=0$. For such states, we generate $U'_X$, $U'_M$ and $U'_Y$ independently conditional on $U$, with
\[
U'_X \sim \text{Bernoulli}\big(\tfrac{1}{2}\big), \qquad
U'_M \sim \text{Bernoulli}\big(\tfrac{9}{10}\big), \qquad
U'_Y\mid U_Y \sim \text{Bernoulli}\big(\tfrac{1}{10} + \tfrac{U_Y}{10}\big).
\]
The success probability of $U'_Y$ is $1/10$ when $U_Y=0$ and $2/10$ when $U_Y=1$, and thus
\[
\mathbb{E}[U'_Y] = \tfrac{1}{2}\cdot\tfrac{1}{10} + \tfrac{1}{2}\cdot\tfrac{2}{10} = \tfrac{3}{20}.
\]
Moreover, $\mathbb{E}[U'_M] = \tfrac{9}{10}$.

\subsection{NDE and BDE}

The NDE is
\[
\text{NDE} = \mathbb{E}[Y_{1,M_0} - Y_{0,M_0}] = \mathbb{E}[(1 + U_M + U_Y) - (U_M + U_Y)] = 1.
\]

For the BDE, the retained pairs satisfy $X(U)=0$, $X(U')=1$ and
$M(U')=M(U)$. Since $X=U_X$, $M(U)=U_M$ and $M(U')=1 \oplus U'_M$, this means
\[
U_X = 0, \qquad U'_X = 1, \qquad U'_M = 1 \oplus U_M.
\]
The retention event involves only $U_X$, $U'_X$, $U_M$ and $U'_M$, and so conditioning on it does not change the distribution of $(U_Y, U'_Y)$. For retained pairs, the factual and alternative outcomes are
\[
Y(U) = U_M + U_Y, \qquad Y(U') = 1 + U_M + U'_Y,
\]
because $M(U') = M(U) = U_M$. Consequently,
\[
\text{BDE}_{P_{BT}} = 1 + \mathbb{E}[U'_Y] - \mathbb{E}[U_Y] = 1 + \tfrac{3}{20} - \tfrac{1}{2} = \tfrac{13}{20} < 1 = \text{NDE}.
\]

\subsection{NIE and BIE}

The NIE is
\[
\text{NIE} = \mathbb{E}[Y_{0,M_1} - Y_{0,M_0}]
= \mathbb{E}[(1 - U_M + U_Y) - (U_M + U_Y)] = \mathbb{E}[1 - 2U_M] = 0.
\]

For the BIE, the retained pairs satisfy $X(U)=0$, $X(U')=0$ and
$M(U')=M_1(U)$. Since $X=U_X$, $M(U') = U'_M$, and
$M_1(U) = 1 - U_M$, this means
\[
U_X = 0, \qquad U'_X = 0, \qquad U'_M = 1 - U_M.
\]
As in the BDE, the retention event selects on $U_M$: conditioning on $R_I(U,U')$ gives $P(U_M=0 \mid R_I(U,U'))=P(U_M'=1 \mid R_I(U,U'))=\tfrac{9}{10}$, and hence
\[
\mathbb{E}[U_M \mid R_I(U,U')]=\tfrac1{10}.
\]
Unlike in the BDE, where the terms in $U_M$ cancel because the mediator is the same in both worlds, this selection now affects the result. The retention event does not involve $U_Y$ or $U'_Y$, and so
$\mathbb{E}[U_Y\mid R_I(U,U')] = \tfrac12$ and $\mathbb{E}[U'_Y\mid R_I(U,U')] = \tfrac{3}{20}$.
For retained pairs, the factual outcome is $Y(U) = U_M + U_Y$, and the
alternative outcome is
\[
Y(U') = U'_M + U'_Y = 1 - U_M + U'_Y.
\]
Consequently,
\[
\text{BIE}_{P_{BT}}
= \mathbb{E}[1 - 2U_M + U'_Y - U_Y  \mid  R_I(U,U')]
= 1 - \tfrac{2}{10} + \tfrac{3}{20} - \tfrac{1}{2}
= \tfrac{9}{20} > 0 = \text{NIE}.
\]

\subsection{TE and BTE}

The TE is
\[
\text{TE} = \mathbb{E}[Y_1 - Y_0]
= \mathbb{E}[(1 + (1 - U_M) + U_Y) - (U_M + U_Y)]
= \mathbb{E}[2 - 2U_M] = 1.
\]

For the BTE, the retained pairs satisfy $X(U)=0$ and $X(U')=1$, i.e., $U_X=0$ and $U'_X=1$. This event does not
involve $U_M$, $U'_M$, $U_Y$ or $U'_Y$, and so conditioning on it does not
change their distributions. For retained pairs,
\[
Y(U) = U_M + U_Y, \qquad
Y(U') = 1 + (1 \oplus U'_M) + U'_Y = 2 - U'_M + U'_Y.
\]
Consequently,
\[
\text{BTE}_{P_{BT}}
= \mathbb{E}[2 - U'_M + U'_Y - U_M - U_Y]
= 2 - \tfrac{9}{10} + \tfrac{3}{20} - \tfrac{1}{2} - \tfrac{1}{2}
= \tfrac{1}{4} < 1 = \text{TE}.
\]

\subsection{Varying $P_{BT}$}

Let us now consider the same backtracking probability distribution as above with the exception that now 
\[
U'_Y\mid U_Y \sim \text{Bernoulli}\big(\tfrac45 + \tfrac{U_Y}{10}\big).
\]
Then, $\mathbb{E}[U'_Y] = \tfrac{17}{20}$ and the retention events and the laws of $U_M$ and $U'_M$ are unchanged. Repeating the calculations above yields BDE=$\tfrac{27}{20}$, BIE=$\tfrac{23}{20}$ and BTE=$\tfrac{19}{20}$.

\section{Proof of Proposition \ref{pro:relation}}\label{app:relation}

\begin{proof}
By Assumptions \ref{ass:marginal} and \ref{ass:uniform}, $U'_M$ is uniform, and the joint distribution of $(U_M,U'_M)$ satisfies
\begin{align*}
P(U_M=0,U'_M=0)&=P(U_M=1,U'_M=1)=\alpha,\\
P(U_M=0,U'_M=1)&=P(U_M=1,U'_M=0)=\tfrac12-\alpha
\end{align*}
for some $\alpha\in[0,\tfrac12]$, because 
\begin{align*}
P(U_M=0)&=P(U_M=0,U'_M=0)+P(U_M=0,U'_M=1)=\tfrac12,\\
P(U_M'=0)&=P(U_M=0,U'_M=0)+P(U_M=1,U'_M=0)=\tfrac12.
\end{align*}

\emph{BTE.} The retention event $R_T(U,U')=\{X(U)=0,X(U')=1\}=\{U_X=0,U'_X=1\}$ involves only the $X$-pair. Moreover, $Y(U')=f_Y(1,f_M(1,U'_M),U'_Y)$ is independent of $(U_X,U_X')$ by Assumption \ref{ass:componentwise}, and $(U'_M,U'_Y)$ has the distribution of $(U_M,U_Y)$ by Assumptions \ref{ass:marginal} and \ref{ass:componentwise}. Thus $\mathbb{E}[Y(U') \mid R_T(U,U')]=\mathbb{E}[Y_1]$. Likewise, $Y(U)=Y_0$ is a function of $(U_M,U_Y)$, which is independent of $(U_X,U_X')$ by Assumption \ref{ass:componentwise}, and so $\mathbb{E}[Y(U) \mid R_T(U,U')]=\mathbb{E}[Y_0]$. Hence $\text{BTE}_{P_{BT}}=\mathbb{E}[Y_1]-\mathbb{E}[Y_0]=\text{TE}$. This step uses only Assumptions \ref{ass:marginal} and \ref{ass:componentwise}.

\emph{BDE.} Since $f_M(x,\cdot)$ is a bijection by Assumption \ref{ass:uniform}, $M(U')=f_M(1,U'_M)$ equals $M(U)=f_M(0,U_M)$ if and only if $U'_M=d(U_M)$, where $d=f_M(1,\cdot)^{-1}\circ f_M(0,\cdot)$ is the identity or the flip. Thus $R_D(U,U')=\{U_X=0,U'_X=1,U'_M=d(U_M)\}$. The probability $P(U_M=u,U'_M=d(u))$ is $\alpha$ for both $u\in\{0,1\}$ if $d$ is the identity, and $\tfrac12-\alpha$ for both if $d$ is the flip. So, given $R_D(U,U')$, $U_M$ is still uniform, and independent of $(U_Y,U'_Y)$ by Assumption \ref{ass:componentwise}. On $R_D(U,U')$, $M(U')=M(U)=f_M(0,U_M)=M_0$, so that $Y(U')=f_Y(1,M_0,U'_Y)$ and $Y(U)=f_Y(0,M_0,U_Y)$. Given $R_D(U,U')$, $(U_M,U'_Y)$ and $(U_M,U_Y)$ both have the distribution of $(U_M,U_Y)$ under $P_U$ by Assumption \ref{ass:marginal}, and so
\[
\mathbb{E}[Y(U') \mid R_D(U,U')]=\mathbb{E}[Y_{1,M_0}],\qquad
\mathbb{E}[Y(U) \mid R_D(U,U')]=\mathbb{E}[Y_{0,M_0}].
\]
Hence $\text{BDE}_{P_{BT}}=\mathbb{E}[Y_{1,M_0}-Y_{0,M_0}]=\text{NDE}$.

\emph{BIE.} Similarly, $M(U')=f_M(0,U'_M)$ equals $M_1=f_M(1,U_M)$ if and only if $U'_M=c(U_M)$, where $c=f_M(0,\cdot)^{-1}\circ f_M(1,\cdot)$ is the identity or the flip. Thus $R_I(U,U')=\{U_X=0,U'_X=0,U'_M=c(U_M)\}$ and, as above, given $R_I(U,U')$, $U_M$ is uniform and independent of $(U_Y,U'_Y)$. On $R_I(U,U')$, $Y(U')=f_Y(0,M_1,U'_Y)$ and $Y(U)=f_Y(0,M_0,U_Y)$, with $M_0=f_M(0,U_M)$ and $M_1=f_M(1,U_M)$. Given $R_I(U,U')$, $(U_M,U'_Y)$ and $(U_M,U_Y)$ both have the distribution of $(U_M,U_Y)$ under $P_U$, and so
\[
\mathbb{E}[Y(U') \mid R_I(U,U')]=\mathbb{E}[Y_{0,M_1}],\qquad
\mathbb{E}[Y(U) \mid R_I(U,U')]=\mathbb{E}[Y_{0,M_0}],
\]
and $\text{BIE}_{P_{BT}}=\mathbb{E}[Y_{0,M_1}-Y_{0,M_0}]=\text{NIE}$.
\end{proof}

\section{Proof of Proposition \ref{pro:unification}}\label{app:unification}

\begin{proof}
For $X(V)=0$, $M(V)=M_0$ and $Y(V)=Y_{0,M_0}$. Moreover, in the first case, $Y(V')=f_Y(1,f_M(1,U_M),U_Y)=Y_1$; in the second, $Y(V')=f_Y(1,M_0,U_Y)=Y_{1,M_0}$; and in the third, since $X(U)=0$, $Y(V')=f_Y(0,M_1,U_Y)=Y_{0,M_1}$. In each case the integrand is a function of $(U_M,U_Y)$ only, which is independent of $U_X$, and so conditioning on $R(V,V')$ does not change the expectation.
\end{proof}

\section{Hybrid Effects}\label{app:hybrid}

The augmented SCM of Section \ref{sec:augmented} allows counterfactual states obtained by hybrid changes, which alter both the circumstances and the regimes. For instance, forcing the mediator to $M(U)$ or $M_1$ while the circumstances change by backtracking, instead of selecting the mediator through the retention set, yields effects that equal the NDE and NIE under Assumptions \ref{ass:marginal} and \ref{ass:componentwise} alone, in contrast to Proposition \ref{pro:relation} where the uniformity of $U_M$ in Assumption \ref{ass:uniform} is required only because the retention sets select on $U_M$. The price is that the hybrid effects override the mediator, and so do not describe a response in which nothing is forced. The following proposition formalizes this.

\begin{myproposition}
Under Assumptions \ref{ass:marginal} and \ref{ass:componentwise},
\[
\text{TE}=\mathbb{E}[Y(V')-Y(V) \mid R(V,V')]
\]
with $R(V,V')=\{X(V)=0,X(V')=1\}$ and $V'=(U',\varnothing,\varnothing)$ and $U'\sim P_{BT}(\cdot \mid u)$,
\[
\text{NDE}=\mathbb{E}[Y(V')-Y(V) \mid R(V,V')]
\]
with $R(V,V')=\{X(V)=0,X(V')=1\}$ and $V'=(U',\varnothing,M(V))$ and $U'\sim P_{BT}(\cdot \mid u)$,
\[
\text{NIE}=\mathbb{E}[Y(V')-Y(V) \mid R(V,V')]
\]
with $R(V,V')=\{X(V)=0,X(V')=0\}$ and $V'=(U',\varnothing,M_1)$ and $U'\sim P_{BT}(\cdot \mid u)$.
\end{myproposition}

\begin{proof}
By Assumption \ref{ass:componentwise}, the pairs $(U_X,U'_X)$, $(U_M,U'_M)$ and $(U_Y,U'_Y)$ are independent. By Assumption \ref{ass:marginal}, $U'$ has
distribution $P_U$, and so each $U'_i$ has the distribution of $U_i$. Hence, since $P_U=\prod_i P(u_i)$ by Assumption \ref{ass:componentwise}, the vectors $(U_M,U_Y)$, $(U_M,U'_Y)$ and $(U'_M,U'_Y)$ all have distribution $P(u_M)P(u_Y)$.

Both retention sets, $\{X(V)=0,X(V')=1\}$ and $\{X(V)=0,X(V')=0\}$, involve only the pair $(U_X,U'_X)$, which is independent of $(U_M,U'_M,U_Y,U'_Y)$. Conditioning on the retention set therefore leaves the distribution of $(U_M,U'_M,U_Y,U'_Y)$ unchanged. Moreover, since $X(V)=0$
on both sets, $M(V)=M_0=f_M(0,U_M)$ and $Y(V)=Y_{0,M_0}$. Hence
\[
\mathbb{E}[Y(V) \mid R(V,V')]=\mathbb{E}[Y_{0,M_0}]
\]
in all three cases, because $Y_{0,M_0}$ is a function of $(U_M,U_Y)$.

\emph{TE.} Let $R(V,V')=\{X(V)=0,X(V')=1\}$ and $V'=(U',\varnothing,\varnothing)$. Then $Y(V')=f_Y(1,f_M(1,U'_M),U'_Y)$ is a function of $(U'_M,U'_Y)$, which has the same distribution as $(U_M,U_Y)$. Hence $\mathbb{E}[Y(V')-Y(V)\mid R(V,V')]=\mathbb{E}[Y_1 - Y_0]=\text{TE}$, since
$Y_0=Y_{0,M_0}$.

\emph{NDE.} Let $R(V,V')=\{X(V)=0,X(V')=1\}$ and $V'=(U',\varnothing,M(V))$. Then $Y(V')=f_Y(1,M_0,U'_Y)$ with $M_0=f_M(0,U_M)$ is a function of $(U_M,U'_Y)$, which has the same distribution as $(U_M,U_Y)$. Hence $\mathbb{E}[Y(V')-Y(V)\mid R(V,V')]=\mathbb{E}[Y_{1,M_0}-Y_{0,M_0}]=\text{NDE}$.

\emph{NIE.} Let $R(V,V')=\{X(V)=0,X(V')=0\}$ and $V'=(U',\varnothing,M_1)$ with
$M_1=f_M(1,U_M)$. Then $Y(V')=f_Y(0,f_M(1,U_M),U'_Y)$ is a function of $(U_M,U'_Y)$, which has the same distribution as $(U_M,U_Y)$. Hence $\mathbb{E}[Y(V')-Y(V) \mid R(V,V')]=\mathbb{E}[Y_{0,M_1}-Y_{0,M_0}]=\text{NIE}$.
\end{proof}

\end{document}